\documentclass{llncs}

\usepackage{amsfonts}
\usepackage{amsmath}
\usepackage{graphicx}
\usepackage{psfrag}

\usepackage{color}

\numberwithin{equation}{section}

\newtheorem{ass}{Assumption}
\newtheorem{defn}{Definition}
\newtheorem{thm}{Theorem}
\newtheorem{prop}{Proposition}
\newtheorem{exam}{Example}
\newtheorem{cthm}{Theorem}

\numberwithin{cthm}{section}

\def\Pr{\mathbb{P}}
\def\Ex{\mathbb{E}}

\def\Rl{\mathbb{R}}
\def\var{{\rm Var\hskip 0.2pt}}
\def\VAR{{\rm VAR\hskip 0.2pt}}

\def\const{{\rm const}}
\def\p{{^\prime}}
\def\t{{^\top}}
\def\d{{\rm d}}
\def\e{{\rm e}}
\def\kr{{\scriptscriptstyle \bullet}}
\def\th{\theta}
\def\eps{\varepsilon}
\def\yo{y_{\rm obs}}
\def\yobs{y_{\rm obs}}
\def\llk{\ell}
\def\llm{\ell_m}
\def\c{c}
\def\mc{\hat{c}}
\def\ic{\hat d}
\def\tr{{\rm tr}}
\def\hopt{{h_\star}}

\def\w{w}
\def\W{W}
\def\Ys{Y^\star}
\def\ker{P}

\def\ths{{\th_\star}}
\def\psis{\psi_\star}
\def\them{\hat \th_m}
\def\ef{\mathcal{F}}
\def\Y{\mathcal{Y}}

\def\grad{\nabla}
\def\hess{\nabla^2}

\def\topr{\xrightarrow{\rm p}}
\def\toas{\xrightarrow{\rm a.s.}}
\def\tod{\xrightarrow{\rm d}}

\begin{document}
\frontmatter          
\pagestyle{headings}  
\addtocmark{MCML} 
\mainmatter              
\title{Adaptive Monte Carlo Maximum Likelihood}
\titlerunning{Adaptive MCML}  
%
\author{B{\l}a\.{z}ej Miasojedow\inst{1} \and Wojciech Niemiro\inst{1,3} \and
Jan Palczewski\inst{2} \and Wojciech Rejchel\inst{3}}
\authorrunning{B. Miasojedow et al.} 
%
\tocauthor{B{\l}a\.{z}ej Miasojedow, Wojciech Niemiro,
Jan Palczewski, Wojciech Rejchel}
\institute{
Faculty of Mathematics, Informatics and Mechanics,
University of Warsaw, Warsaw, Poland\\
\email{W.Miasojedow@mimuw.edu.pl\\ W.Niemiro@mimuw.edu.pl}
\and
School of Mathematics, University of Leeds, Leeds, UK\\
\email{J.Palczewski@leeds.ac.uk} 
\and
Faculty of Mathematics and Computer Science,
Nicolaus Copernicus University, Toru\'n, Poland\\
\email{wrejchel@gmail.com}
}

\maketitle              

\begin{abstract}
We consider Monte Carlo approximations to the maximum likelihood estimator in  models
with intractable norming constants. This paper deals with adaptive Monte Carlo algorithms,
which adjust control parameters in the course of simulation. We examine asymptotics of
adaptive importance sampling and a new algorithm, which uses resampling and MCMC. This algorithm is
designed to reduce problems with degeneracy of importance weights.
Our analysis is based on martingale limit theorems. We also describe how adaptive
maximization algorithms of Newton-Raphson type can be combined with the resampling techniques.
The paper includes results of a small scale simulation study in which we compare the performance
of adaptive and non-adaptive Monte Carlo maximum likelihood algorithms.
\end{abstract}

\section{Introduction}\label{Sec:Introduction}
Maximum likelihood (ML) is a well-known and often used method in estimation
of parameters in statistical models. However, for many complex models exact calculation of such estimators
is very difficult or impossible. Such problems arise if considered densities are known only up to intractable
norming constants, for instance in Markov random fields or spatial statistics. The wide range of applications
of models with unknown norming constants is discussed e.g.\  in \cite{MPRB2006}.
Methods proposed to overcome the problems with computing ML estimates in such models include, among others,
maximum pseudolikelihood \cite{Besag1974}, ``coding method'' \cite{HuWu1998} and
Monte Carlo maximum likelihood (MCML) \cite{GeyerThom1992}, \cite{WuHu1997}, \cite{HuWu1998}, \cite{Imput2010}.
In our paper we focus on MCML.
\goodbreak

In influential papers \cite{GeyerThom1992}, \cite{Geyer1994} the authors prove consistency and asymptotic normality of
MCML estimators. To improve the performance of MCML, one can  adjust control parameters in the course of simulation.
This leads to adaptive MCML algorithms. We generalize the results of the last mentioned papers first to
an adaptive version of importance sampling and then to a more complicated adaptive algorithm which uses resampling
and Markov chain Monte Carlo (MCMC) \cite{ISREMC}. Our analysis is asymptotic and it is  based on the martingale structure
of the estimates.
The main motivating examples are the autologistic model (with or without covariates) and its applications to spatial
statistics as described e.g.\ in \cite{HuWu1998} and the autonormal model \cite{Pettitt2002}.

\section{Adaptive Importance Sampling}\label{AIS}

Denote by $f_\theta$, $\theta \in \Theta$, a family of unnormalized densities on space $\Y$.  A dominating measure with respect to which
these densities are defined is denoted for simplicity by $\d y$. Let $\yo$ be an observation.
We intend to find the maximizer $\ths$ of the log-likelihood
\begin{equation}\nonumber
\llk(\theta) = \log f_\theta(\yo) - \log \c(\theta),
\end{equation}
where $\c(\theta)$ is the normalizing constant. We consider the situation  where this constant,
\begin{equation}\nonumber
 \c(\theta)=\int_\Y f_\theta(y)\d y,
\end{equation}
is \textit{unknown and numerically intractable}.
It is approximated with  Monte Carlo simulation, resulting in
\begin{equation}\label{MCloglik}
\llm(\theta) =  \log f_\theta(\yo) - \log \mc_m(\theta),
\end{equation}
where $\mc_m(\th)$ is a Monte Carlo (MC) estimate of $\c(\th)$. The classical importance sampling (IS)
estimate is of the form
\begin{equation}\label{ISestim}
 \mc_m(\theta) = \frac{1}{m} \sum_{j=1}^m \frac{f_\theta(Y_j)}{h(Y_j)},
\end{equation}
where $Y_1,\ldots,Y_m$ are i.i.d.\ samples from an instrumental density $h$.
Clearly, an optimal choice of $h$ depends on the maximizer $\ths$ of $\llk$, so we should be able to improve our
initial guess about $h$ while the simulation progresses. This is the idea behind \textit{adaptive importance sampling}
(AIS). A discussion on the choice of instrumental density is deferred to subsequent subsections.
Let us describe an adaptive algorithm in the following form, suitable for further generalizations.
Consider a parametric family $h_\psi$, $\psi\in\Psi$ of instrumental densities.
\bigskip\goodbreak

\begin{center}
 \tt  Algorithm AdapIS
\end{center}
\begin{enumerate}\setlength{\parskip}{0pt}\tt
\item Set an initial value of $\psi_1$, $m=1$, $\mc_0(\theta) \equiv 0$.
\item Draw $Y_m \sim h_{\psi_m}$.
\item Update the approximation of $c(\theta)$:
\begin{equation}\nonumber
 \mc_{m}(\theta) = \frac{m-1}{m} \mc_{m-1}(\theta) + \frac{1}{m} \frac{f_\th(Y_m)}{h_{\psi_m}(Y_m)}.
\end{equation}
\item Update $\psi$: choose $\psi_{m+1}$ based on the history of the simulation.
\item $m=m+1$; go to 2.
\end{enumerate}
\goodbreak

At the output of this algorithm we obtain an AIS estimate
\begin{equation}\label{AISestim}
 \mc_m(\theta) = \frac{1}{m} \sum_{j=1}^m \frac{f_\theta(Y_j)}{h_{\psi_j}(Y_j)}.
\end{equation}
The samples $Y_j$ are neither independent nor have the same distribution. However (\ref{AISestim})
has a nice \textit{martingale} structure.
If we put
\begin{equation}\nonumber
 \ef_{m} = \sigma \left\{ Y_j,\psi_j: j\leq m\right\}
\end{equation}
then $\psi_{m+1}$ is $\ef_{m}$-measurable. The well-known property of unbiasedness of IS
implies that
\begin{equation}\label{MG0}
\Ex\left( \frac{f_\theta(Y_{m+1})}{h_{\psi_{m+1}}(Y_{m+1})}\Big|\ef_{m}\right)=c(\th).
\end{equation}
In other words, ${f_\theta(Y_{m})}/{h_{\psi_{m}}(Y_{m})}-c(\th)$
are martingale differences (MGD),  for every fixed $\th$.
\subsection{Hypo-convergence of $\llm$ and consistency of $\them$}
\newcommand{\ee}{\Ex}%
\newcommand{\en}{\mathcal{N}}%
\newcommand{\ve}{\varepsilon}%

In this subsection we make the following assumptions.

\begin{ass}\label{asJanka}
For any $\theta \in \Theta$
\begin{equation}\nonumber
   \sup_\psi \int \frac{f_\theta(y)^2}{h_{\psi}(y)}\d y<\infty.
\end{equation}
\end{ass}

\begin{ass}\label{ContinuityTheta}
 The mapping $\theta \mapsto f_\theta(y)$ is continuous for each fixed $y$.
\end{ass}
Assumption \ref{asJanka} implies that for any $\theta$, {there is a constant $M_\theta < \infty$ such that for all $j$,
\begin{equation}\nonumber
\Ex \left(\bigg(\frac{f_\theta(Y_j)}{h_{\psi_j}(Y_j)}\bigg)^2 \bigg|\ef_{j-1}\right)\le M_\theta,\quad \text{a.s.},
\end{equation}
because $Y_j \sim h_{\psi_j}$. Note that Assumption \ref{asJanka}
is trivially true if the mapping $y \mapsto f_\theta(y) / h_{\psi}(y)$ is uniformly bounded for $\theta\in\Theta$, $\psi\in\Psi$.
Recall also that
\begin{equation}\nonumber
m(\mc_m(\th)-\c(\th)) = \sum_{j=1}^m \Big( \frac{f_\theta(Y_j)}{h_{\psi_j}(Y_j)} - \c(\theta) \Big)
\end{equation}
is a zero-mean martingale.
Under Assumption \ref{asJanka}, for a fixed $\theta \in \Theta$, we have
$\mc_m(\theta) \to \c(\theta)$ a.s.\  by the SLLN for martingales (see Theorem \ref{SLLN}, Appendix \ref{Martingales}), so
$\llm(\theta) \to \llk(\theta)$ a.s. This is, however, insufficient to guarantee the convergence
of maximum likelihood estimates $\them$ (maximizers of $\llm$) to $\ths$. Under our assumptions we can
prove hypo-convergence of the log-likelihood approximations.

\begin{defn}
A sequence of functions $g_m$ epi-converges to $g$ if for any $x$ we have
\begin{align*}
&\sup_{B \in N(x)} \limsup_{m \to \infty} \inf_{y \in B} g_m(y) \le g(x),\\
&\sup_{B \in N(x)} \liminf_{m \to \infty} \inf_{y \in B} g_m(y) \ge g(x),
\end{align*}
where $N(x)$ is a family of all (open) neighbourhoods of $x$.

A sequence of functions $g_m$ hypo-converges to $g$ if $(-g_m)$ epi-converges to $(-g)$.
\end{defn}
An equivalent definition of epi-convergence is in the following theorem:
\begin{thm}(\cite[Proposition 7.2]{Rockafellar2009})\label{thm:epi-conv}
$g_m$ epi-converges to $g$ iff at every point $x$
\begin{align*}
&\limsup_{m \to \infty} g_m(x_m) \le g(x), \qquad \text{for some sequence $x_m \to x$,}\\
&\liminf_{m \to \infty} g_m(x_m) \ge g(x), \qquad \text{for every sequence $x_m \to x$.}
\end{align*}
\end{thm}

As a corollary to this theorem comes the proposition that will be used to prove convergence of $\them$,
the maximizer of $\llm$, to $\ths$ (see, also, \cite[Theorem 1.10]{Attouch1984}).
\begin{prop}\label{EpiConsistency}
Assume that $g_m$ epi-converges to $g$, $x_m \to x$ and $g_m(x_m) - \inf g_m \to 0$. Then $g(x) = \inf_y g(y) = \lim_{m \to \infty} g_m(x_m)$.
\end{prop}
\begin{proof}
(We will use Theorem \ref{thm:epi-conv} many times.) Let $y_m$ be a sequence converging to $x$ and such that $\limsup_{m \to \infty} g_m(y_m) \le g(x)$ (such sequence $y_m$ exists). This implies that $\limsup_{m \to \infty} \inf g_m \le g(x)$.
On the other hand, $g(x) \le \liminf_{m \to \infty} g_m(x_m) =\liminf_{m \to \infty} \inf g_m$, {where the equality follows from the second assumption on $x_m$}.
Summarizing, $g(x) = \lim_{m \to \infty} \inf g_m = \lim_{m \to \infty} g_m(x_m)$. In particular, $\inf g \le \lim_{m \to \infty} \inf g_m$.

Take any $\ve > 0$ and let $x_\ve$ be such that $g(x_\ve) \le \inf g + \ve$. There exists a sequence $y_m$ converging to $x_\ve$ such that $g(x_\ve) \ge \limsup_{m \to \infty} g_m(y_m) \ge \limsup_{m \to \infty} \inf g_m$, hence $\lim_{m \to \infty} \inf g_m \le \inf g + \ve$. By arbitrariness of $\ve>0$ we obtain $\lim_{m \to \infty} \inf g_m \le \inf g$. This completes the proof.
\end{proof}

\begin{thm}\label{thm:hypoconv}
If Assumptions \ref{asJanka} and \ref{ContinuityTheta} are satisfied, then $\llm$ hypo-converges to $\ell$
almost surely.
\end{thm}
\begin{proof}
The proof is similar to the proof of Theorem 1 in \cite{Geyer1994}. We have to prove that $\mc_m$ epi-converges to
$\c$. Fix $\theta \in \Theta$.

Step 1: For any $B \in N(\theta)$, we have
\begin{equation}\label{eqn:epi-conv_step1}
\liminf_{m \to \infty} \inf_{\varphi \in B} \mc_m(\varphi) \ge \int \inf_{\varphi \in B} f_\varphi(y) \d y =: \underline{c}(B).
\end{equation}
Indeed,
\begin{align*}
\inf_{\varphi \in B} \mc_m(\phi)
&= \inf_{\varphi \in B} \frac{1}{m} \sum_{j=1}^m \frac{f_\varphi(Y_j)}{h_{\psi_j}(Y_j)}
\ge \frac{1}{m} \sum_{j=1}^m \inf_{\varphi \in B} \frac{f_\varphi(Y_j)}{h_{\psi_j}(Y_j)}\\
&=  \frac{1}{m} \sum_{j=1}^m \Big( \inf_{\varphi \in B} \frac{f_\varphi(Y_j)}{h_{\psi_j}(Y_j)} - \underline \c(B) \Big) + \underline \c(B).
\end{align*}
The sum is that of martingale differences, so assuming that {there is $M < \infty$ such that}
\begin{equation}\nonumber
{\sup_j \ee \bigg(\Big(\inf_{\varphi \in B} \frac{f_\varphi(Y_j)}{h_{\psi_j}(Y_j)} - \underline \c(B) \Big)^2 \bigg| \ef_{j-1}\bigg) \le M}
\end{equation}
the SLLN implies \eqref{eqn:epi-conv_step1}. We have the following estimates: 
\begin{align*}
&\ee \bigg(\Big(\inf_{\varphi \in B} \frac{f_\varphi(Y_j)}{h_{\psi_j}(Y_j)} - \underline \c(B) \Big)^2 \bigg| \ef_{j-1}\bigg)
=
\var \bigg(\inf_{\varphi \in B} \frac{f_\varphi(Y_j)}{h_{\psi_j}(Y_j)} \bigg| \ef_{j-1}\bigg)\\
&\le
\ee \bigg(\Big(\inf_{\varphi \in B} \frac{f_\varphi(Y_j)}{h_{\psi_j}(Y_j)} \Big)^2 \bigg| \ef_{j-1}\bigg)
\le
\ee \bigg(\Big(\frac{f_\theta(Y_j)}{h_{\psi_j}(Y_j)} \Big)^2 \bigg| \ef_{j-1}\bigg)
\le M_\theta,
\end{align*}
where the last inequality is by Assumption \ref{asJanka}.

Step 2: We shall prove that $\sup_{B \in N(\theta)} \liminf_{m \to \infty} \inf_{\varphi \in B} \mc_m(\phi) \ge \c(\theta)$.

The left-hand side is bounded from below by
$\sup_{B \in N(\theta)} \underline \c(B)$. Further, we have
\begin{equation}\nonumber
\sup_{B \in N(\theta)} \underline \c(B) \ge \lim_{\delta \downarrow 0} \underline \c(B(\theta, \delta))
= \int \lim_{\delta \downarrow 0} \inf_{\varphi \in B(\theta, \delta)} f_\varphi(y) \d y = \int f_\theta(y) \d y =
\c(\theta),
\end{equation}
where the first equality follows from the dominated convergence theorem (the dominator is $f_\theta$) and the last --
from the Assumption \ref{ContinuityTheta}.

Step 3: We have
\begin{equation}\nonumber
\sup_{B \in N(\theta)} \limsup_{m \to \infty} \inf_{\varphi \in B} \mc_m(\varphi)
\le \sup_{B \in N(\theta)} \inf_{\varphi \in B} \limsup_{m \to \infty} \mc_m(\varphi)
=  \sup_{B \in N(\theta)} \inf_{\varphi \in B} \c(\varphi) \le \c(\theta).
\end{equation}
Hence, $\sup_{B \in N(\theta)} \limsup_{m \to \infty} \inf_{\varphi \in B} \mc_m(\varphi) \le \c(\theta)$.
\end{proof}

{Note that almost sure convergence in the next Proposition corresponds to the randomness introduced by \texttt{AdapIS} and
$\yobs$ is fixed throughout this paper.}

\begin{prop}
If Assumptions \ref{asJanka} and \ref{ContinuityTheta} hold, $\ths$ is the unique maximizer of $\llk$ and
sequence $(\them)$ {(where $\them$ maximizes $\llm$)} is almost surely bounded then $\them\to \ths$ almost surely.
\end{prop}
\begin{proof} As we have already mentioned, by SLLN for martingales, $\llm(\th)\to \llk(\th)$, pointwise.
Hypo-convergence of $\llm$ to $\ell$ implies, by Proposition \ref{EpiConsistency}, that the maximizers of $\llm$
have accumulation points that are the  maximizers of $\ell$. If $\ell$ has a unique maximizer $\ths$ then any
convergent subsequence of $\them$, maximizers of $\llm$, converges to $\ths$.
The conclusion follows immediately. 
\end{proof}

Of course, it is not easy to show boundedness of $\them$ in concrete examples. In the next section
we will prove consistency of $\them$ in models where log-likelihoods and their estimates are concave.

\subsection{Central Limit Theorem for Adaptive Importance Sampling}

Let $\them$ be a maximizer of $\llm$, i.e.\ the AIS estimate of the likelihood given by \eqref{MCloglik} with \eqref{AISestim}.
We assume that $\ths$ is a unique maximizer of $\ell$.
For asymptotic normality of $\them$, we will need the following assumptions.

\begin{ass}\label{Derivatives}
 First and second order derivatives of $f_\th$ with respect to $\th$
(denoted by $\grad f_\th$ and $\hess f_\th$) exist in a neighbourhood of $\ths$ and we have
\begin{equation}\nonumber
 \grad \c(\theta)=\int \grad f_\theta(y)\d y,\qquad \hess\c(\theta)=\int \hess f_\theta(y)\d y.
\end{equation}
\end{ass}

\begin{ass}\label{ConsistencySqrt}
 $\them=\ths+O_{\rm p}(1/\sqrt{m})$.
\end{ass}

\begin{ass}\label{Dpositive}
 Matrix $D=\hess \ell(\ths)$ is negative definite.
\end{ass}

\begin{ass}\label{ContinuityPsi}
For every $y$, function $\psi\mapsto h_\psi(y)$ is continuous and $h_\psi(y)>0$.
\end{ass}

\begin{ass}\label{Diminishing} For some $\psi_\star$ we have
 $\psi_m\to \psi_\star$ almost surely.
\end{ass}

\begin{ass}\label{Momentdiff} There exists a nonnegative function $g$ such that $\int g(y)\d y<\infty$ and the inequalities
\begin{equation}\nonumber
\begin{split}
&\sup_{\psi} \frac{f_\ths(y)^{2+\alpha}}{h_{\psi}(y)^{1+\alpha}} \leq g(y),  \quad
 \sup_{\psi} \frac{|\grad f_\ths(y)|^{2+\alpha}}{h_{\psi}(y)^{1+\alpha}} \leq g(y),\\
 &\sup_{\psi} \frac{\Vert \hess f_\ths(y)\Vert^{1+\alpha}}{h_{\psi}(y)^{\alpha}} \leq g(y)
\end{split}
\end{equation}
are fulfilled for some $\alpha>0$ and also for $\alpha=0$.
 \end{ass}

 \begin{ass}\label{ASE} Functions $\hess\llm(\th)$ are asymptotically stochastically  equicontionuous at $\ths$, i.e.\
for every $\eps>0$ there exists $\delta>0$ such that
\begin{equation}\nonumber
 \limsup_{m\to\infty} \Pr\left( \sup_{\vert \th-\ths\vert \leq \delta} \Vert \hess \llm(\th)-\hess \llm(\ths) \Vert >\eps\right)=0.
\end{equation}
 \end{ass}
Let us begin with some comments on these assumptions and note simple facts which follow from them.
Assumption \ref{Derivatives} is a standard regularity condition.
It implies that a martingale property similar to \eqref{MG0} holds also for the gradients and hessians:
\begin{equation}\label{MG12}
\Ex\left( \frac{\grad f_\theta(Y_{m+1})}{h_{\psi_{m+1}}(Y_{m+1})}\Big|\ef_{m}\right)=\grad c(\th),\quad
\Ex\left( \frac{\hess f_\theta(Y_{m+1})}{h_{\psi_{m+1}}(Y_{m+1})}\Big|\ef_{m}\right)=\hess c(\th).
\end{equation}
Assumption \ref{ConsistencySqrt} stipulates square root consistency of $\them$. It is automatically
fulfilled if $\llm$ is concave, in particular for exponential families.
Assumption \ref{Diminishing} combined with \ref{ContinuityPsi} is a ``diminishing adaptation'' condition. It may be ensured by an appropriately
specifying step 4 of \texttt{AdapIS}.
The next assumptions are not easy to verify in general, but they are satisfied
for exponential families on finite spaces, in particular for our ``motivating example'': autologistic model.
Let us also note that our Assumption \ref{ASE} plays a similar role to Assumption (f) in \cite[Thm. 7]{Geyer1994}.

Assumption \ref{Momentdiff} together with \eqref{MG0} and \eqref{MG12} allows us to apply SLLN for martingales
in a form given in Theorem \ref{SLLN}, Appendix \ref{Martingales}. Indeed,
${f_\ths(Y_{m})}/{h_{\psi_{m}}(Y_{m})}-c(\ths)$, ${\grad f_\ths(Y_{m})}/{h_{\psi_{m}}(Y_{m})}-\grad c(\ths)$
and ${\hess f_\ths(Y_{m})}/{h_{\psi_{m}}(Y_{m})}-\hess c(\ths)$
are MGDs with bounded moments of order $1+\alpha>1$. It follows that, almost surely,
 \begin{equation}\label{Consistency}
  \mc_m(\ths)\to \c(\ths),\qquad \grad \mc_m(\ths)\to \grad \c(\ths),\qquad \hess\mc_m(\ths)\to  \hess\c(\ths).
 \end{equation}

Now we are in a position to state the main result of this section.

\begin{thm}\label{AsNormAIS}
If Assumptions 
\ref{Derivatives}-\ref{ASE} hold
then
\begin{equation}\nonumber
\sqrt{m}\left(\them -\ths\right) \to \mathcal{N} (0,D^{-1}VD^{-1}) \quad \text{in distribution},
 \end{equation}
where $D=\hess\ell(\ths)$ and
 \begin{equation}\nonumber
  V = \frac{1}{\c(\ths)^2} \VAR_{Y \sim h_{\psis}} \left[\frac{\grad f_{\ths}(Y)}{h_{\psis}(Y)}-\frac{\grad c(\ths)}{c(\ths)}\dfrac{ f_{\ths}(Y)}{h_{\psis}(Y)}\right],
 \end{equation}
{where $\psis$ is defined in Assumption \ref{Diminishing}.}
\end{thm}

\begin{proof}
It is well-known (see
\cite[Theorem VII.5]{Pollard1984})
that we need to prove
\begin{equation}\label{as_norm1}
\sqrt{m}\grad \llm(\ths)\tod \mathcal{N} (0,V)
\end{equation}
and that for every $M>0$, the following holds:
\begin{equation}\label{sup1}
\begin{split}
\sup_{|\th-\ths|\leq {M}/{\sqrt{m}}} & m\Big|\llm(\th)-\llm(\ths)\\
       &- (\th-\ths)\t \grad\llm(\ths)-\frac{1}{2}(\th-\ths)\t D(\th-\ths)\Big| \topr 0.
\end{split}
\end{equation}

First we show \eqref{as_norm1}. Since $\grad \llm(\th)=\grad f_\th(\yo)/ f_\th(\yo) -\grad\mc_m(\th)/\mc_m(\th)$ and
 $\grad \ell(\ths)=\grad f_\ths(\yo)/f_\ths(\yo) -\grad\c(\ths)/\c(\ths)=0$, we obtain that
\begin{equation}\label{NumDenom}
\grad \llm(\ths) = \frac{\grad\c(\ths)}{\c(\ths)}-\frac{\grad\mc_m(\ths)}{\mc_m(\ths)}=
\frac{\dfrac{\grad\c(\ths)}{\c(\ths)} \mc_m(\ths)-\grad\mc_m(\ths)}{ \mc_m(\ths) }.
\end{equation}
The denominator in the above expression converges to $\c(\ths)$ in probability, by \eqref{Consistency}. In view of Slutski's
theorem, to prove \eqref{as_norm1} it is enough to show asymptotic normality of the numerator.
We can write
\begin{equation}\nonumber
 \dfrac{\grad\c(\ths)}{\c(\ths)} \mc_m(\ths)-\grad\mc_m(\ths)=-\frac{1}{m}\sum_{j=1}^m \xi_j,
\end{equation}
where we use the notation
\begin{equation}\nonumber
\xi_j= \dfrac{\grad f_{\ths}(Y_j)}{h_{\psi_j}(Y_j)}-
                        \dfrac{\grad c(\ths)}{c(\ths)} \dfrac{f_{\ths}(Y_j)}{h_{\psi_j}(Y_j)}.
\end{equation}
Now note that $\xi_j$ are martingale differences by \eqref{MG0} and \eqref{MG12}. Moreover,
\begin{equation}\nonumber
\begin{split}
\Ex\Big( \xi_j \xi_j^T | \ef_{j-1} \Big)=\int &\left(\dfrac{\grad f_{\ths}(y)}{h_{\psi_j}(y)}-
                        \dfrac{\grad c(\ths)}{c(\ths)} \dfrac{f_{\ths}(y)}{h_{\psi_j}(y)}\right)\\
                     &\left(\dfrac{\grad f_{\ths}(y)}{h_{\psi_j}(y)}-
                        \dfrac{\grad c(\ths)}{c(\ths)} \dfrac{f_{\ths}(y)}{h_{\psi_j}(y)}\right)\t h_{\psi_j}(y)\d y,
\end{split}
\end{equation}
so Assumptions \ref{ContinuityPsi} and \ref{Diminishing} via dominated convergence and Assumption \ref{Momentdiff} (with $\alpha=0$ in the exponent) entail
\begin{equation}\nonumber
\Ex\Big( \xi_j \xi_j^T | \ef_{j-1} \Big)\toas \c(\ths)^2V.
\end{equation}
Now we use Assumption \ref{Momentdiff} (with $\alpha>0$ in the exponent) to infer the Lyapunov-type condition
\begin{equation}\nonumber
\Ex\Big( |\xi_j|^{2+\alpha} | \ef_{j-1} \Big)\leq \const\cdot \int g(y)\d y<\infty.
\end{equation}
The last two displayed formulas are sufficient for a martingale CLT (Theorem \ref{CLT}, Appendix \ref{Martingales}). We conclude that
\begin{equation}\nonumber
\frac{1}{\sqrt{m}} \sum_{j=1}^m \xi_j\xi_j\t \tod \mathcal{N} (0,\c(\ths)^2 V),
\end{equation}
hence the proof of \eqref{as_norm1} is complete.

Now we proceed to a proof of \eqref{sup1}. By Taylor expansion,
\begin{equation}\nonumber
 \llm(\th)-\llm(\ths)- (\th-\ths)\t \grad\llm(\ths)=\frac{1}{2}(\th-\ths)\t \hess\llm(\tilde \th)(\th-\ths)
\end{equation}
for some $\tilde\th\in[\th,\ths]$. Consequently, the LHS of \eqref{sup1} is
\begin{equation}\nonumber
\begin{split}
 &\leq \sup_{\substack{|\th-\ths|\leq {M}/\sqrt{m}\\ \tilde\th\in[\th,\ths]}}
              m \left|\frac{1}{2}(\th-\ths)\t\left( \hess\llm(\tilde\th)- \hess\ell(\ths)\right)(\th-\ths)\right|\\
 &\leq \frac{M^2}{2} \: \sup_{|\th-\ths|\leq {M}/\sqrt{m}}
               \left\Vert  \hess\llm(\th)- \hess\ell(\ths)\right\Vert\\
 &\leq \frac{M^2}{2} \: \sup_{|\th-\ths|\leq {M}/\sqrt{m}}
               \left\Vert \hess\llm(\th)- \hess\llm(\ths)\right\Vert
                  + \frac{M^2}{2} \: \left\Vert\hess\llm(\ths)-\hess\ell(\ths)\right\Vert.\\
\end{split}
\end{equation}
The first term above goes to zero in probability by Assumption \ref{ASE}. The second term also goes to zero because
\begin{equation}\nonumber
\begin{split}
 &\hess \llm(\ths)-\hess\ell(\ths)=\hess \log\c(\ths)-\hess\log \mc_m(\ths)\\
	&=\frac{\hess\c(\ths)}{\c(\ths)}-\frac{\grad\c(\ths)}{\c(\ths)}\frac{\grad\c(\ths)\t}{\c(\ths)}
           - \frac{\hess\mc_m(\ths)}{\mc_m(\ths)}+\frac{\grad\mc_m(\ths)}{\mc_m(\ths)}\frac{\grad\mc_m(\ths)\t}{\mc_m(\ths)}\topr 0,\\
\end{split}
\end{equation}
in view of \eqref{Consistency}. Therefore \eqref{sup1} holds and the proof is complete.
\end{proof}
\goodbreak

\subsection{Optimal importance distribution}

\newcommand\ag{\eta}

We advocate adaptation to improve the choice of instrumental distribution $h$. But which $h$ is the best?
If we use (non-adaptive) importance sampling with instrumental distribution $h$ then the maximizer $\them$ of
the MC likelihood approximation has asymptotic normal distribution, namely
$\sqrt{m}\left(\them -\ths\right) \tod \mathcal{N} (0,D^{-1}VD^{-1})$, ($m \rightarrow \infty$) with
 \begin{equation}\nonumber
  V = \frac{1}{\c(\ths)^2} \VAR_{Y \sim h}
\left[\frac{\grad f_{\ths}(Y)}{h(Y)}-\frac{\grad c(\ths)}{c(\ths)}\dfrac{ f_{\ths}(Y)}{h(Y)}\right].
 \end{equation}
This fact is well-known \cite{Geyer1994} and  is a special case of Theorem \ref{AsNormAIS}.
Since the asymptotic distribution is multidimensional its dispersion can be measured in various ways, e.g.,  though the determinant,
the maximum eigenvalue or the trace of the covariance matrix.
We examine the trace which equals to the asymptotic mean square error of the MCML approximation (the asymptotic bias is nil).
Notice that
\begin{equation}\nonumber
 \c(\ths)^2 V =\VAR_{Y \sim h}  \frac{\ag(Y)}{h(Y)} =\Ex_{Y \sim h}  \frac{ \ag(Y) \ag(Y)\t}{h(Y)^2}=\int \frac{ \ag(y) \ag(y)\t }{h(y)}\d y,
\end{equation}
where
\begin{equation}\nonumber
\ag(y)= \grad f_{\ths}(y)-\frac{\grad c(\ths)}{c(\ths)}{ f_{\ths}(y)}.
\end{equation}
Since $\tr \left[D^{-1} \ag(y) \ag(y)\t D^{-1}\right]= (D^{-1} \ag(y))\t D^{-1} \ag(y)=\vert D^{-1} \ag(y)\vert^2$,
the minimization of $\tr(D^{-1} V D^{-1})$ is equivalent to
\begin{equation}\nonumber
\int \frac{\vert D^{-1} \ag(y)\vert^2}{h(y)}\d y\to \min,
\end{equation}
subject to $h\geq 0$ and $\int h=1$. By Schwarz inequality we have 
\begin{equation}\nonumber
\begin{split}
 \int \frac{\vert D^{-1} \ag(y)\vert^2}{h(y)}\d y&= \int \left(\frac{\vert D^{-1} \ag(y)\vert}{\sqrt{h(y)}}\right)^2\d y
                                           \int \left(\sqrt{h(y)}\right)^2\d y\\
                                        &\geq \left( \int \frac{\vert D^{-1} \ag(y)\vert}{\sqrt{h(y)}} \sqrt{h(y)}\d y\right)^{2}
                                        = \left( \int \vert D^{-1} \eta(y)\vert\d y\right)^{2},
\end{split}
\end{equation}
with equality only for $h(y)\propto |D^{-1} \ag(y)|$.  The optimum choice of $h$ is therefore
\begin{equation}\label{Optimal}
\hopt(y)\propto \left|D^{-1} \left( \grad f_\ths(y)-\dfrac{\grad \c(\ths)}{\c(\ths)} f_\ths(y)\right) \right|.
\end{equation}
Unfortunately, this optimality result is chiefly of theoretical importance, because it is not clear
how to sample from $\hopt$ and how to compute the norming constant for this distribution. This might well
be even more difficult than computing $\c(\th)$.

The following example shows some intuitive meaning of \eqref{Optimal}. It is a purely ``toy example'' because
the simple analitical formulas exist for $\c(\th)$ and $\ths$ while MC is considered only for illustration.

\begin{exam}
Consider a binomial distribution on $\Y=\{0,1,\ldots,n\}$ given by $\pi_\th(y)=\binom{n}{y}p^y(1-p)^{n-y}$.
Parametrize the model with the log-odds-ratio $\th=\log p/(1-p)$, absorb the $\binom{n}{y}$ factor into
the measure $\d y$ to get the standard exponenial family form with
\begin{equation}\nonumber
 f_\th(y)=\e^{\th y} \quad\text{and}\quad \c(\th)=\sum_{y=0}^n \binom{n}{y} \e^{\th y}=(1+\e^\th)^n.
\end{equation}
Taking into account the facts that ${\grad \c(\ths)}/{\c(\ths)}=\yo$ and $\grad f_\th(y)=y\e^{\th y}$ we obtain
that \eqref{Optimal} becomes $\hopt(y)\propto|y-\yo|\e^{\th y}$ (factor $D^{-1}$ is a scalar so can be omitted). In other words, the optimum instrumental distribution
for AIS MCML, expressed in terms of $p=\e^{\th}/(1+\e^{\th})$ is
\begin{equation}\nonumber
 \Pr_{Y\sim \hopt}(Y=y)\propto  \binom{n}{y} | y-\yo| p^y(1-p)^{n-y}.
\end{equation}
\end{exam}

\section{Generalized adaptive scheme}\label{Generalized}

Importance sampling, even in its adaptive version (AIS), suffers from the degeneracy of weights.
To compute the importance weights ${f_\theta(Y_m)}/{h_{\psi_m}(Y_m)}$
we have to know norming constants for every $h_{\psi_m}$
(or at least their ratios). This requirement severly restricts our choice of the
family of instrumental densities $h_\psi$. Available instrumental densities are far from
$\hopt$ and far from $f_\th/\c(\th)$.  Consequently the weights tend to degenerate (most of them are practically zero, while
a few are very large). This effectively makes AIS in its basic form
impractical. To obtain  practically applicable algorithms, we can generalize AIS as follows.
In the same situation as in Section \ref{AIS}, instead of the AIS estimate given by \eqref{AISestim}, we
consider a more general Monte Carlo estimate of $\c(\th)$ of the form
\begin{equation}\label{MeanMCestim}
 \mc_m(\theta) = \frac{1}{m} \sum_{j=1}^m \ic(\th,\psi_j),
\end{equation}
where the summands $\ic(\th,\psi_j)$ are computed by an MC method to be specified later. For now let us just assume
that this method depends on a control parameter $\psi$ which may change at each step.
A general adaptive algorithm is the following:
\goodbreak

\begin{center}
 \tt  Algorithm AdapMCML
\end{center}
\begin{enumerate}\setlength{\parskip}{0pt}\tt
\item Set an initial value of $\psi_1$, $m=1$, $\mc_0(\theta) \equiv 0$.
\item Compute an ``incremental estimate'' $\ic(\th,\psi_m)$.
\item Update the approximation of $\mc_m(\theta)$:
\begin{equation}\nonumber
 \mc_{m}(\theta) = \frac{m-1}{m} \mc_{m-1}(\theta) + \frac{1}{m} \ic(\th,\psi_m).
\end{equation}
\item Update $\psi$: choose $\psi_{m+1}$ based on the history of the simulation.
\item $m=m+1$; go to 2.
\end{enumerate}
\texttt{AdapIS} in Section \ref{AIS} is a special case of \texttt{AdapMCML} which is obtained by letting
$\ic(\theta,\psi_m) = {f_\theta(Y_m)}/{h_{\psi_m}(Y_m)}$.
\goodbreak

\subsection{Variance reduction via resampling and MCMC}

The key property of the AIS exploited in Section \ref{AIS}
is the martingale structure implied by \eqref{MG0} and \eqref{MG12}.
The main asymptotic results generalize if
\textit{given $\psi$, the estimates of $\c(\th)$ and its derivatives
are conditionally unbiased}.
We propose an algorithm for computing $\ic$ in \eqref{MeanMCestim} which has the unbiasedness
property and is more efficient than simple AIS. To some extent it is a remedy for the problem of weight degeneracy
{and reduces the variance of Monte Carlo approximations}.
As before, consider a family of ``instrumental densities'' $h_\psi$. Assume they are
properly normalized ($\int h_\psi=1$) and the control parameter $\psi$ belongs
the same space as the parameter of interest $\th$  ($\Psi=\Theta$).  Further assume that for every $\psi$ we have at our disposal
a Markov kernel $\ker_\psi$ on $\Y$ which preserves distribution $\pi_\psi=f_\psi/c(\psi)$, i.e.\
$ f_\psi(y)\d y=\int f_\psi(y\p)\ker_\psi(y\p,\d y)\d y\p$. Let us fix $\psi.$
This is a setup in which we can apply the following importance sampling-\-resampling algorithm \texttt{ISReMC}:
\goodbreak

\begin{center}
 \tt  Algorithm ISReMC
\end{center}
\begin{enumerate}\setlength{\parskip}{0pt}\tt
\item Sample $Y_1,\ldots,Y_l\sim h_\psi$.
\item Compute the importance weights $\W_i=\w(Y_i)=\dfrac{f_\psi(Y_i)}{h_\psi(Y_i)}$ and put $\W_\kr=\sum_{i=1}^l\W_i$.
\item Sample $\Ys_1,\ldots,\Ys_r\sim\sum_{i=1}^l \delta_{Y_i}(\cdot)\W_i/\W_\kr$ [Discrete distribution with mass
$\W_i/\W_\kr$ at point
 $Y_i$].
\item  For $k=1,\ldots,r$ generate a Markov chain trajectory, starting from $\Ys_k$ and using kernel $\ker_\psi$:
\begin{equation}\nonumber
 \Ys_k=Y_k^0, Y_k^1,\ldots,Y_k^s,Y_k^{s+1},\ldots, Y_k^{s+n}.
\end{equation}
\item[] Compute $\ic(\th,\psi)$ given by
\begin{equation}\label{Increm}
  \ic(\th,\psi) =
\dfrac{\W_\kr}{l}\frac{1}{r}\sum_{k=1}^r \frac{1}{n}\sum_{u=s+1}^{s+n} \frac{f_\th(Y^{u}_k)}{f_\psi(Y^{u}_k)}.
\end{equation}
\end{enumerate}
\goodbreak

This algorithm combines the idea of resampling (borrowed from sequential MC; steps 2 and 3)
with computing ergodic averages in multistart MCMC
(step 4; notice that $s$ is a burn-in and $n$ is the actual used sample size for a single MCMC run, repeated  $r$ times).
More details about \texttt{ISReMC} are in
\cite{ISREMC}. In our context it is sufficient to note the following key property of this algorithm.
\begin{lemma}\label{lem:isremc} If $\ic(\th,\psi)$ is the output of \texttt{IReMC} then for every $\th$ and every $\psi$,
 \begin{equation}\nonumber
 \Ex \ic(\th,\psi)=\c(\th).
\end{equation}
If Assumption \ref{Derivatives} holds then also
\begin{equation}\nonumber
 \Ex \grad \ic(\th,\psi)=\grad \c(\th),\qquad \Ex\hess \ic(\th,\psi)=\hess \c(\th).
\end{equation}
\end{lemma}
\begin{proof}
We can express function $\c(\th)$ and its derivatives
as ``unnormalized expectations'' with respect to the probability {distribution with density} $\pi_\psi=f_\psi/c(\psi)$:
\begin{equation}\nonumber
\begin{split}
 \c(\th)&=\Ex_{Y\sim \pi_\psi} \frac{f_\th(Y)}{f_\psi(Y)}\c(\psi),\\
 \grad \c(\th)&=\Ex_{Y\sim \pi_\psi} \frac{\grad f_\th(Y)}{f_\psi(Y)}\c(\psi),
 \qquad \hess\c(\th)=\Ex_{Y\sim \pi_\psi} \frac{\hess f_\th(Y)}{f_\psi(Y)}\c(\psi).
\end{split}
 \end{equation}
Let us focus on $\Ex \ic(\th,\psi)$. Write
\begin{equation}
 a(y)=\Ex\left(\frac{1}{n}\sum_{u=s+1}^{s+n}\frac{f_\th(Y^{u})}{f_\psi(Y^{u})}\Bigg|Y^{0}=y\right)
\end{equation}
for the expectation of a \textit{single} MCMC estimate started at $Y^0=y$. Kernel $\ker_\psi$ preserves $\pi_\psi$ by assumption,
therefore $\Ex_{Y\sim \pi_\psi} a(Y)=\Ex_{Y\sim \pi_\psi} f_\th(Y)/f_\psi(Y)=c(\th)/c(\psi)$. Put differently,
$\int a(y)f_\psi(y)\d y=c(\th)$.

We make a simple observation that
\begin{equation}\nonumber
  \Ex \left(\ic(\th,\psi) \big|Y_1,\ldots,Y_l,\Ys_1,\ldots,\Ys_r \right)=\frac{\W_\kr}{l}\frac{1}{r}\sum_{k=1}^r a(\Ys_k).
\end{equation}
This conditional expectation takes into account only randomness of the MCMC estimate in step 4 of the algorithm.
Now we consecutively ``drop the conditions'':
\begin{equation}\nonumber
  \Ex \left(\ic(\th,\psi) \big|Y_1,\ldots,Y_l\right)=\frac{\W_\kr}{l}\sum_{i=1}^l a(Y_i)\frac{W_i}{W_\kr}
                                                    =\frac{1}{l}\sum_{i=1}^l a(Y_i)W_i.
\end{equation}
{The expectation above takes into account the randomness of the resampling in step 3.}
Finally, since $Y_i\sim h_\psi$ in step 1, we have
\begin{equation}\nonumber
\begin{split}
  \Ex\ic(\th,\psi)&=\Ex a(Y_i)W_i
                            =\Ex_{Y\sim h_\psi} a(Y)\frac{f_\psi(Y)}{h_\psi(Y)}\\
                  &=\int a(y)f_\psi(y)\d y=c(\th).
\end{split}
  \end{equation}
This ends the proof for $\ic$. Exactly the same  argument applies to $\grad \ic$ and $\hess \ic$.
\end{proof}

We can embed the unbiased estimators produced by \texttt{ISReMC} in
our general adaptive scheme \texttt{AdapMCML}. At each step $m$ of the adaptive algorithm, we
have a new control parameter $\psi_m$.  We generate a sample from $h_{\psi_m}$, compute weights, resample and run MCMC using
$\psi_m$. Note that the whole sampling scheme at stage $m$ (including computation of weights) depends on $\psi_m$ but not on $\th$.
In the adaptive algorithm random variable $\psi_{m+1}$ is $\ef_{m}$ measurable,
where $\ef_{m}$ is the history of simulation up to stage $m$.
Therefore the sequence of incremental estimates $\ic(\th,\psi_m)$ satisfies,
for every $\th\in\Theta$,
\begin{equation}\label{MGcond0}
 \Ex(\ic(\th,\psi_{m+1})|\ef_{m})=\c(\th).
\end{equation}
Moreover, first and second derivatives exist and
\begin{equation}\label{MGcond12}
 \Ex(\grad \ic(\th,\psi_{m+1})|\ef_{m})=\grad \c(\th),\qquad \Ex(\hess \ic(\th,\psi_{m+1})|\ef_{m})=\hess \c(\th).
\end{equation}
Formulas \eqref{MGcond0} and \eqref{MGcond12} are analogues of \eqref{MG0} and \eqref{MG12}.
\goodbreak


\subsection{Asymptotics of adaptive MCML}

In this subsection we restrict our considerations to \textit{exponential families on finite spaces}.
This will allow us to prove main results without formulating complicated technical assumptions
(integrability  conditions analoguous to Assumption \ref{Momentdiff} would be cumbersome and difficult to verify).
Some models with important applications, such as \textit{autologistic} one, satisfy the assumptions below.

\begin{ass}\label{ExpoFamily}
Let
\begin{equation}\nonumber
 f_\th(y)=\exp[\th\t t(y)],
\end{equation}
where $t(y)\in \Rl^d$ is the vector of sufficient statistics
and $\th\in\Theta=\Rl^d$. Assume that $y$ belongs to a finite space $\Y$.
\end{ass}
Now, since $\Y$ is finite  (although possibly very large),
\begin{equation}\nonumber
 \c(\th)=\sum_y \exp[\th\t t(y)].
\end{equation}
Note that Assumption \ref{Derivatives} is automatically satisfied.

\begin{ass}\label{CompactPsi}
Control parameters $\psi$ belong to a compact set $\Psi\subset\Rl^d$.
\end{ass}

We consider algorithm \texttt{AdapMCML} with incremental estimates $\ic$ produced by \texttt{ISReMC}.
The likelihood ratio in \eqref{Increm} and its derivatives assume the following form:
\begin{equation}\label{dcExp}%
\begin{split}
   \frac{f_\th(Y)}{f_\psi(Y)}&=\exp[(\th-\psi)\t t(Y)], \\
    \frac{\grad f_\th(Y)}{f_\psi(Y)}&=t(Y)\exp[(\th-\psi)\t t(Y)],\\
     \frac{\hess f_\th(Y)}{f_\psi(Y)}&= t(Y)t(Y)\t\exp[(\th-\psi)\t t(Y)]
\end{split}
\end{equation}
(the derivatives are with respect to $\th$, with $\psi$ fixed).
Assumptions \ref{ExpoFamily} and \ref{CompactPsi} together with Assumption \ref{ContinuityPsi} imply that
$\ic(\th,\psi_j)$ are uniformly bounded, if $\th$ belongs to a compact set.
Indeed, the importance weights $W_i$ in \eqref{Increm} are uniformly bounded by Assumptions \ref{CompactPsi} and \ref{ContinuityPsi}.
Formula \eqref{dcExp} shows that the ratios $f_\th(y)/f_{\psi_j}(y)= \exp[(\th-\psi_j)\t t(y)]$ are also uniformly bounded
for $\psi_j$ and $\th$ belonging to bounded sets. Since the statistics $t(y)$ are bounded, the same argument shows
that $\grad\ic(\th,\psi_j)$ and $\hess\ic(\th,\psi_j)$ are uniformly bounded, too.

For exponential families, $\log \c(\th)$ and   $\log \mc_m(\th)$ are convex functions.
It is a well known property of exponential family that $\hess \log \c(\th)=\VAR_{Y\sim\pi_\th} t(Y)$ and thus it is a nonnegative
definite matrix. A closer look at $\mc_m(\th)$ reveals that $\hess \log \mc_m(\th)$ is also a variance-covariance matrix
with respect to some discrete distribution. Indeed, it is enough to note that $\mc_m(\th)$ is of the form
\begin{equation}\nonumber
   \mc_m(\th) =
\sum_{j,k,u} \exp[\th\t t_{j,k,u}]a_{j,k,u},
\end{equation}
for some $t_{j,k,u}\in \Rl^d$ and $a_{j,k,u}>0$ (although if \texttt{ISReMC} within \texttt{AdapMCML} is used to produce $\mc_m(\th)$ then
$t_{j,k,u}$ and $a_{j,k,u}$ are quite complicated random variables depending on $\psi_j$).

Let $\them$ be a {maximizer} of $\llm(\th)=\th\t t(\yo)-\log\mc_m(\th)$ and assume
that $\ths$ is the unique {maximizer} of $\llk(\th)=\th\t t(\yo)-\log \c(\th)$.

\begin{prop}
If Assumptions \ref{ExpoFamily}, \ref{CompactPsi} and \ref{ContinuityPsi}  hold, then
$\them\to \ths$ almost surely.
\end{prop}
\begin{proof} Boundedness of $\ic(\th,\psi_m)$ for a fixed $\th$ together with \eqref{MGcond0} implies that $\ic(\th,\psi_m)-\c(\th)$ is a
bounded sequence of martingale differences. It satisfies the assumptions of SLLN for martingales in Appendix \ref{Martingales}.
Therefore  $\mc_m(\th)\to \c(\th)$. Consequently, we also have $\llm(\th)\to \llk(\th)$, pointwise.
Pointwise convergence of convex functions
implies uniform convergence on compact sets \cite[Th. 10.8]{Rockafellar1970}. The  conclusion follows immediately.
\end{proof}

\begin{thm}\label{AsNormGeneral} If Assumptions \ref{Dpositive}, \ref{ContinuityPsi}, \ref{Diminishing}, \ref{ExpoFamily} and \ref{CompactPsi}
hold, then
 \begin{equation}\nonumber
  \sqrt{m}(\them-\ths)\to \mathcal{N} (0,D^{-1}V D^{-1}) \text{ in distribution},
 \end{equation}
where $D=\hess \llk(\ths)$ and
\begin{equation}\nonumber
   V=\frac{1}{\c(\ths)^{2}}\VAR \left[\grad \ic(\ths,\psis)- \frac{\grad\c(\ths)}{\c(\ths)}\ic(\ths,\psis)\right],
\end{equation}
where $\ic(\ths,\psis)$ is a result of the IS / Resampling algorithm \texttt{ISReMC}, described in the previous subsection,  with $\psi=\psis$ and $\th=\ths$.
\end{thm}
Note that $\ic(\ths,\psis)$ is a purely imaginary object, being a result of an algorithm initialized at a ``limiting instrumental parameter'' $\psis$
and evaluated at the ``true MLE'' $\ths$, both unknown. It is introduced only to concisely describe the variance/covariance matrix $V$.
Note also that ${\grad\c(\ths)}/{\c(\ths)}$ is equal to $t(\yo)$, the \textit{observed} value of the sufficient statistic.

\begin{proof}[of Theorem \ref{AsNormGeneral}] The proof is similar to that of Theorem \ref{AsNormAIS}, so we will not repeat all the details.
The key argument is again based on  SLLN and CLT for martingales (see Appendix \ref{Martingales}).
In the present situation we have more complicated estimators $\ic(\th,\psi_j)$ than in Theorem \ref{AsNormAIS}. They are now given by \eqref{Increm}.
On the other hand, we work under the  assumption that $f_\th$ is an exponential family on a finite state space $\Y$. This implies that
conditions \eqref{MGcond0} and \eqref{MGcond12} are fulfilled and the martingale differences therein are uniformly bounded
(for any fixed $\th$ and also for $\th$ running through a compact set).
Concavity of $\llm(\th)$ and $\llk(\th)$ further simplifies the argumentation.

As in the proof of Theorem \ref{AsNormAIS}, we claim that \eqref{as_norm1} and \eqref{sup1} hold. The first of these conditions, \eqref{as_norm1},
is justified exactly in the same way: by applying the CLT to the numerator and SLLN to the denominator of \eqref{NumDenom}.
Now, we consider martingale differences given by
\begin{equation}\nonumber
\xi_j= \grad\ic(\ths,\psi_j)-
                        \dfrac{\grad c(\ths)}{c(\ths)} \ic(\ths,\psi_j).
\end{equation}
It follows from the discussion preceding the theorem that $\xi_j$ are uniformly bounded, so the CLT can be  applied. Similarly,
SLLN can be applied to  $\ic(\ths,\psi_j)-c(\ths)$.

Assumption  \eqref{ASE} holds because third order derivatives of $\log \mc_m(\th)$ are uniformly bounded in the neighbouhood of $\ths$.
This allows us to infer condition \eqref{sup1} in the same way as in the proof of Theorem \ref{AsNormAIS}.

Note also that we do not need Assumption \ref{ConsistencySqrt}.
To deduce the conclusion of the theorem from  \eqref{sup1} we have to know that $\them$ is square-root consistent. But this follows from the facts that
$\llm(\th)$ is concave and the maximizer of the quadratic function $- (\th-\ths)\t \grad\llm(\ths)-\frac{1}{2}(\th-\ths)\t D(\th-\ths)$
in \eqref{sup1} is square-root consistent by \eqref{as_norm1}.
\end{proof}


\section{Simulation results}

In a series of small scale simulation experiments, we compare two algorithms. The first one, used as a ``Benchmark'' is a non-adaptive MCML.
The other is  \texttt{AdapMCML} which uses \texttt{ISReMC} estimators, as described in Section \ref{Generalized}. Synthetic data used in our study
are generated from autologistic model, described below. Both algorithms use Gibbs Sampler (GS) as an MCMC subroutine and both
use Newton-Raphson iterations to maximize MC log-likelihood approximations.

\subsection{Non-adaptive and adaptive Newton-Raphson-type algorithms}

Well-known Newton-Raphson (NR) method in our context updates points $\th_m$ approximating maximum of the log-likelihood as follows:
\begin{equation}\nonumber
 \th_{m+1}= \th_m +\hess \ell_m(\th_m)^{-1} \grad\ell_m(\th_m),
\end{equation}
where $\ell_m$ is given by \eqref{MCloglik}.

\textbf{ Non-adaptive} algorithms are obtained when some fixed value of the ``instrumental parameter''
is used to produce MC samples. Below we recall a basic version of such an algorithm, proposed be Geyer \cite{Geyer1994} and examined e.g.\ in
\cite{HuWu1998}. If we consider an exponenial family given by Assumption \ref{ExpoFamily}, then
$\llm(\th) =  \th\t t(\yo)  - \log \mc_m(\th)$. Let $\psi$ be fixed and $Y_{0},Y_{1},\ldots,Y_{s},\ldots,Y_{s+m}$ be  samples
approximately drawn from  distribution $\pi_\psi\propto f_\psi$. 
In practice an MCMC method is applied to produce such samples, $s$ stands for a burn-in.
In all our experiments the MCMC method is  a deterministic scan Gibbs Sampler (GS).
Now, we let
\begin{equation}\nonumber
\begin{split}
  \mc_m(\th) \propto\frac{1}{m}\sum_{u=s+1}^{s+m} \exp[(\th-\psi)\t t(Y_{u})].\\
\end{split}
\end{equation}
Consequently, if $\omega_u(\th)=\exp[(\th-\psi)\t t(Y_{u})]$ and $\omega_{\kr}(\th)=\sum_{u=s+1}^{s+m}\omega_u(\th)$, then
the derivatives of the log-likelihood are expressed via weighted moments,
\begin{equation}\nonumber
\begin{split}
\grad \llm(\th) &=t(\yo)-\overline{t(Y)},\quad \overline{t(Y)}=\frac{1}{\omega_{\kr}(\th)}\sum_{u=s+1}^{s+m} \omega_u(\th)t(Y_{u}),\\
\hess \llm(\th) &=-\frac{1}{\omega_{\kr}(\th)}\sum_{u=s+1}^{s+m} \omega_u(\th)(t(Y_{u}-\overline{t(Y)})(t(Y_{u})-\overline{t(Y)})\t.\\
\end{split}
\end{equation}

 \textbf{The adaptive} algorithm uses $\mc_m(\th)$ given by \eqref{MeanMCestim}, with summands $\ic(\th,\psi_j)$ computed by \texttt{ISReMC},
exactly as described in Section \ref{Generalized}. The MCMC method imbedded in \texttt{ISReMC} is GS, the same as in the non-adaptive algorithm.
Importance sampling distribution $h_\psi$ in steps 1 and 2 of \texttt{ISReMC} is pseudo-likelihood, described by formula \eqref{PseudoLik} in the next subsection.
Computation of $\psi_{m+1}$ in step 4 of \texttt{AdapMCML} uses one NR iteration:
{$\psi_{m+1}=\psi_m +\hess \ell_m(\psi_m)^{-1} \grad\ell_m(\psi_m)$,
where $\ell_m$ is given by \eqref{MCloglik} with $\mc_m$ produced by \texttt{AdapMCML}.}


\subsection{Methodology of simulations}

For our experiments we have chosen  the autologistic model, one of chief motivating examples for MCML.
It is given by a probability distribution on $\Y=\{0,1\}^{d\times d}$ proportional to
\begin{equation}\notag 
f_{\th}(y)=\exp\left(
  \th_{0}\sum_{r}y^{(r)} + \th_{1}\sum_{r\sim s} y^{(r)}y^{(s)}\right),
\end{equation}
where $r\sim s$ means that two points $r$ and $s$ in the $d\times d$ lattice are neighbours.
The pseudo-likelihood $h_\psi$ is given by
\begin{equation}\label{PseudoLik}
h_{\psi}(y)\propto\prod_{r} \exp\left( \th_{0}y^{(r)}+
  \th_{1}\sum_{s:r\sim s}y^{(r)}\yo^{(s)}\right).
\end{equation}

In our study we considered lattices of dimension
$d=10$ and $d=15$. 
The values of sufficient statistics
$T=\left(\sum_{r}y^{(r)},\sum_{r\sim s} y^{(r)}y^{(s)}\right)$,
exact ML estimators $\ths$ and maxima of the log-likelihoods are in the Tables 1 and 2 below.
We report results of several repeated runs of a ``benchmark'' non-adaptive algorithm and our adaptive
algorithm. The initial points are 1) the maximum pseudo-likelihood (MPL) estimate, denoted by $\hat\th$ (also included in the tables) and
2) point $(0,0)$. Number of runs is $100$ for $d=10$ and $25$ for $d=15$. Below we describe the parameters
and results of these simulations. Note that we have chosen parameters for both algorithms in such a way which allows for a ``fair comparison'', that
is the amount of computations and number of required samples are similar for the benchmark and adaptive algorithms.

\textbf{For $d=10$:} In benchmark MCML, we used $1000$ burn-in and $39\,000$ collected realisations of the Gibbs sampler;
then $20$ iterations of Newton-Raphson were applied. \texttt{AdapMCML} had $20$ iterations; parameters within \texttt{ISReMC} were
$l=1000$, $r=1$, $s=100$, $n=900$.

$$
\begin{tabular}{|c|c|c|c|}
\multicolumn{4}{c}{Table 1.}\\
\hline
Statistic $T$        & ML   $\ths$       & Log-Lik $\ell(\ths)$ & MPL  $\hat\th$\\
\hline
$(59,74)$  &  $(-1.21,0.75)$ & $ -15.889991$ & $(-1.07,0.66)$\\
\hline
\end{tabular}
$$
\bigskip

The results are shown in Figures 1 and 2.

\begin{figure}\label{BoxD10}

\centering

 \includegraphics[width=\textwidth]{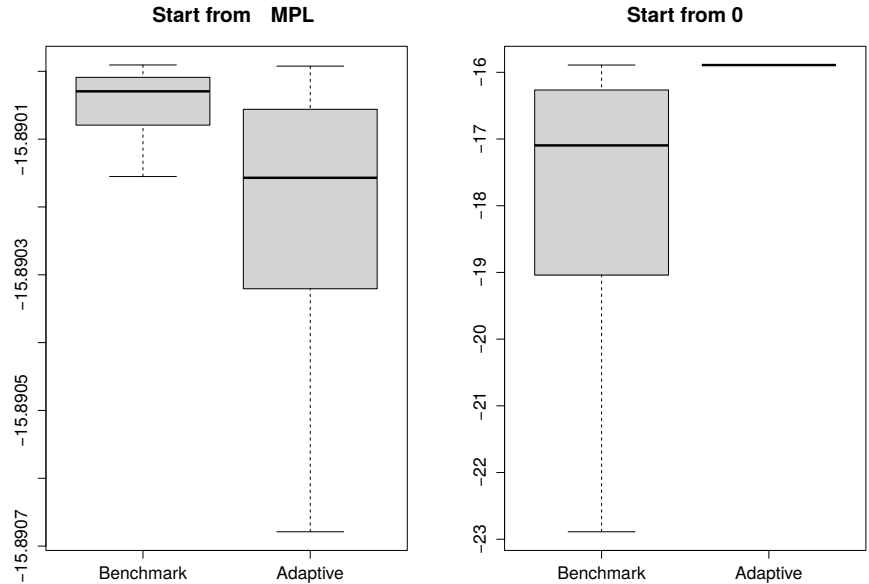}
\caption{Log-likelihood at the output of MCML algorithms; $d=10$; $100$ runs.}
 \end{figure}

\bigskip\goodbreak

\begin{figure}\label{ContourD10}
\centering
 \includegraphics[width=\textwidth]{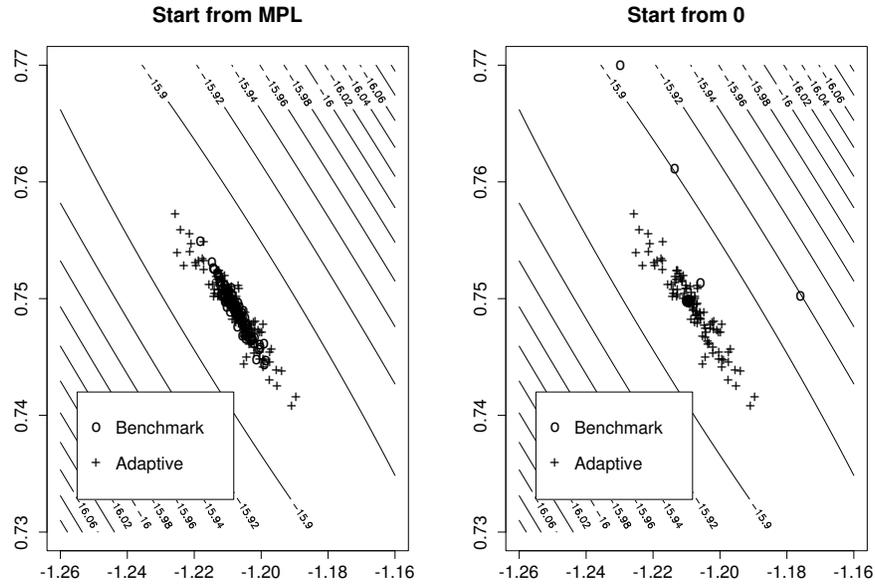}
\caption{Output of MCML algorithms;  $d=10$; 100 repetitions.}
 \end{figure}

\bigskip\goodbreak

\textbf{For $d=15$:} In benchmark MCML, we used $10\,000$ burn-in and $290\,000$ collected realisations of the Gibbs sampler;
then $10$ iterations of Newton-Raphson were applied. \texttt{AdapMCML} had $10$ iterations; parameters within \texttt{ISReMC} were
$l=10\,000$, $r=1$, $s=1000$, $n=19\,000$.

$$
\begin{tabular}{|c|c|c|c|}
\multicolumn{3}{c}{Table 2.}\\
\hline
Statistic $T$        & ML   $\ths$       & Log-Lik $\ell(\ths)$ & MPL  $\hat\th$\\
\hline
$(142,180)$  &  $(-0.46,0.43)$ & $12.080011$ & $(-0.57,0.54)$ \\
\hline
\end{tabular}
$$
\bigskip

The results are shown in Figures 3 and 4. The benchmark algorithm started from 0 for $d=15$  failed, so only the results for the adaptive algorithm are given in the right parts of Figures 3 and 4.

\begin{figure}[ht!]\label{BoxD15}
\centering
 \includegraphics[width=\textwidth]{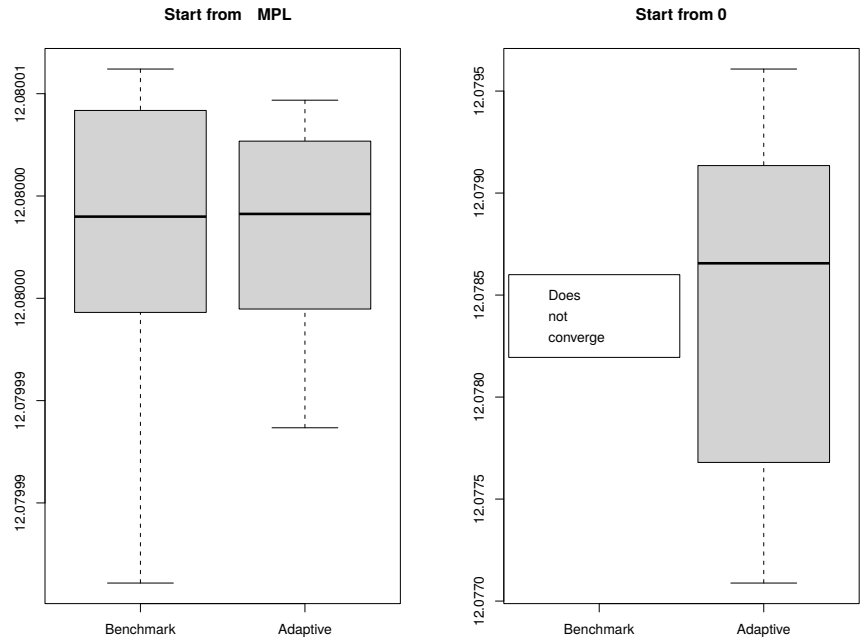}
\caption{Log-likelihood at the output of MCML algorithms; $d=15$; 25 repetitions.}
 \end{figure}

\begin{figure}\label{ContourD15}
\centering
 \includegraphics[width=\textwidth]{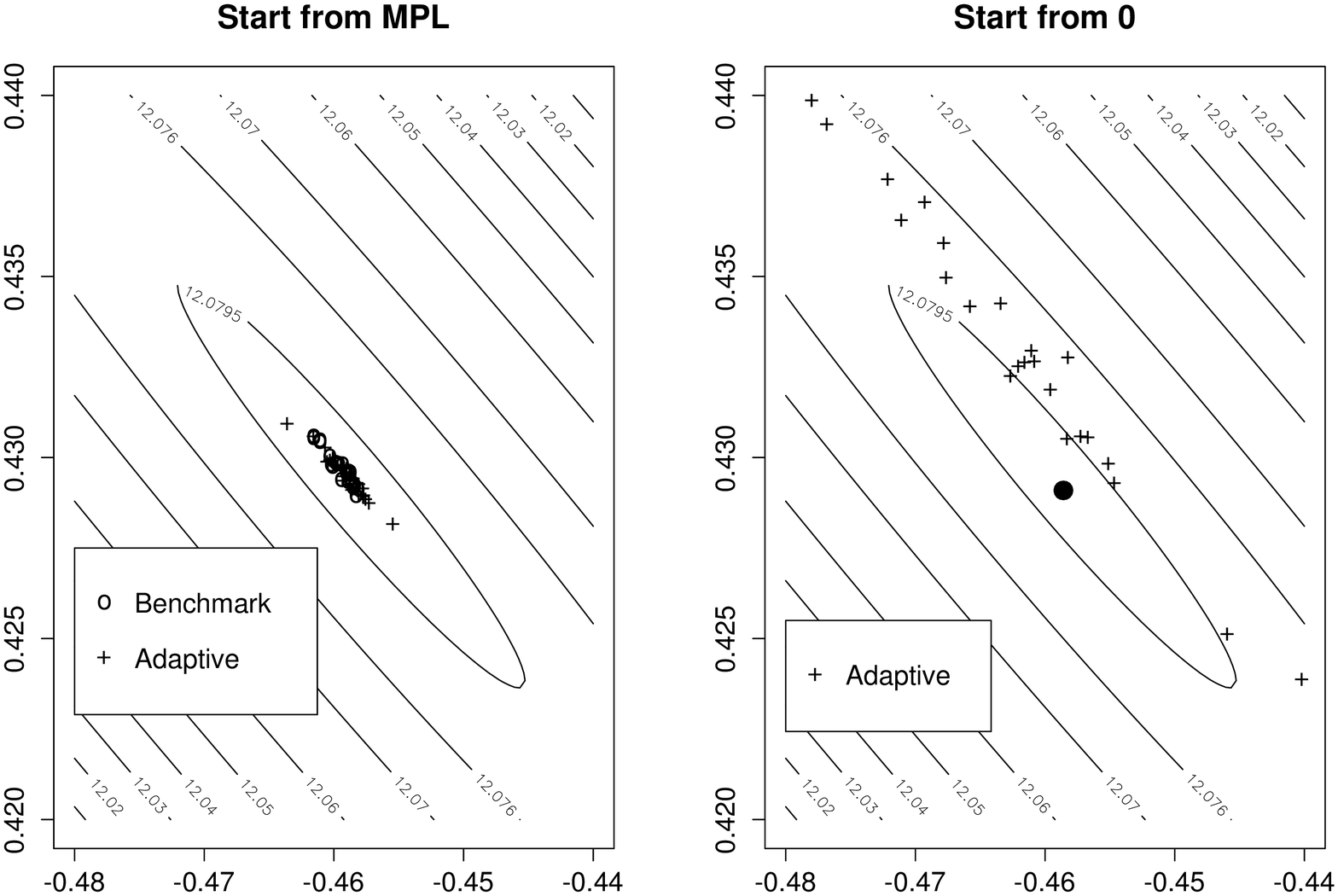}
\caption{Output of MCML algorithms;  $d=15$; 25 repetitions.}
 \end{figure}

\subsection{Conclusions}

The results of our simulations allow to draw only some preliminary conclusions, because the range of experiments was limited.
However, some general conclusions can be rather safely formulated. The performance of the benchmark, non-adaptive algorithm crucially
depends on the choice of starting point. It yields quite satisfactory results, if started sufficiently close tho the maximum likelihood,
for example from the maximum pseudo-likelihood estimate. Our adaptive algorithm is much more robust and stable in this respect.
If started from a good initial point, it may give slightly worse results than the  benchmark, but still is satisfactory (see Fig. \ref{BoxD10}). However, when the maximum pseudo-likelihood estimate is not that close to the maximum likelihood point, the adaptive algorithm yields an estimate with a lower variance (see Fig. \ref{BoxD15}). When started at a point distant from the maximum likelihood, such as $0$, it works much better than a non-adaptive algorithm. Thus the algorithm
proposed in our paper can be considered as more universal and robust alternative to a standard MCML estimator.

Finally let us remark that there are several possibilities of improving our adaptive algorithm. Some heuristically justified modifications
seem to converge faster and be more stable than the basic version which we described. Modifications can exploit the idea of resampling in a
different way and  reweigh past samples in subsequent steps. Algorithms based on stochastic approximation, for example such as that proposed in \cite{Y1988},
can probably be improved by using Newton-Raphson method instead of simple gradient descent. However, theoretical analysis of such modified algorithms
becomes more difficult and rigorous theorems about them are not available yet. This is why we decided not to include these modified
algorithms in this paper. Further research is needed to bridge a gap between practice and theory of MCML.
\bigskip

\textbf{Acknowledgement.} This work was partially supported by Polish National Science Center No. N N201 608 740.

\newpage

%
%

\appendix

\section{Appendix: martingale limit theorems}\label{Martingales}

For completeness, we cite the following martingale central limit theorem (CLT):
\begin{cthm}(\cite[Theorem 2.5]{Helland1982})\label{CLT}
Let $X_n = \xi_1 + \ldots + \xi_n$ be a mean-zero (vector valued) martingale. If there exists a symmetric positive definite matrix $V$ such that
\begin{equation}\label{Variances}
  \frac{1}{n} \sum_{j=1}^n \ee \Big(  \xi_j \xi_j^T | \ef_{j-1} \Big) \topr V,
\end{equation}
\begin{equation}\label{Lindeberg}
  \frac{1}{n} \sum_{j=1}^n \ee \Big(  \xi_j \xi_j^T \mathbf{1}_{|\xi_j| > \ve \sqrt{n}}\, | \ef_{j-1} \Big) \topr 0
    \quad\text{ for each } \ve > 0,
\end{equation}
then
\begin{equation}\nonumber
\frac{X_n}{\sqrt{n}} \tod N(0, V).
\end{equation}
\end{cthm}
The Lindeberg condition \eqref{Lindeberg} can be replaced by a stronger Lyapunov condition
\begin{equation}\label{Lyapunov}
  \frac{1}{n} \sum_{j=1}^n \ee \Big(  |\xi_j|^{2+\alpha} | \ef_{j-1} \Big)\leq M
   \quad\text{ for some } \alpha > 0 \text{ and }M<\infty.
\end{equation}

A simple consequence of \cite[Theorem 2.18]{HaHey1980} (see also \cite{Chow1967}) is the following strong law of large numbers (SLLN).
\begin{cthm}\label{SLLN}
Let $X_n = \xi_1 + \ldots + \xi_n$ be a mean-zero martingale. If
\begin{equation}\nonumber
 \sup_j \Ex \Big(|\xi_j|^{1+\alpha} | \ef_{j-1} \Big)\leq M \quad\text{ for some } \alpha > 0 \text{ and }M<\infty
\end{equation}
then
\begin{equation}\nonumber
\frac{X_n}{n} \toas 0.
\end{equation}
\end{cthm}


\begin{thebibliography}{}
%

\bibitem{Attouch1984} Attouch, H. (1984). \textit{Variational Convergence of Functions and Operators}, Pitman.

\bibitem{Besag1974} Besag J. (1974). Spatial interaction and the statistical analysis of lattice systems. \textit{J.R. Statist. Soc. B}, 36, 192--236.

\bibitem{Chow1967} Chow, Y.S. (1967) On a strong law of large numbers for martingales. \textit{Ann. Math. Statist.}, 38,
 610.

\bibitem{GeyerThom1992} Geyer C.J. and Thompson E.A. (1992). Constrained Monte Carlo maximum
likelihood for dependent data. \textit{J.R. Statist. Soc. B}, 54, 657--699.


\bibitem{Geyer1994} Geyer, C.J. (1994). On the Convergence of Monte Carlo Maximum Likelihood Calculations,
\textit{J. R. Statist. Soc. B}, 56, 261--274.

\bibitem{HaHey1980} Hall, P., Heyde, C.C. (1980). \textit{Martinagale Limit Theory and Its Application}.
 Academic Press.

\bibitem{ISREMC} Miasojedow, B., Niemiro, W. (2014). Debiasing MCMC via importance sampling-resampling. \textit{In preparation.}

\bibitem{Helland1982} Helland, I.S. (1982). Central Limit Theorems for Martingales with Discrete or Continuous Time.
\textit{Scand J. Statist.}, 9, 79--94.

\bibitem{HuWu1998} Huffer F.W., Wu H. (1998). Markov chain Monte Carlo for autologistic
regression models with application to the distribution of plant species. \textit{Biometrics}, 54, 509--524.

\bibitem{MPRB2006} M\o ller B.J., Pettitt A.N., Reeves R. and Berthelsen, K.K. (2006).
An efficient Markov chain Monte Carlo method for distributions with intractable normalising constants.
 \textit{Biometrika},  93, 451--458.

\bibitem{Pettitt2002} Pettitt, A.N., Weir, I.S. and Hart, A.G. (2002). A conditional autoregressive Gaussian
process for irregularly spaced multivariate data with application to modelling large
sets of binary data. \textit{Statistics and Computing} 12, 353--367.


\bibitem{Pollard1984} Pollard D. (1984). \textit{Convergence of stochastic processes}, Springer,
New York.

\bibitem{Rockafellar1970} Rockafellar, R.T. (1970). \textit{Convex Analysis}. Princeton University Press, Princeton.

\bibitem{Rockafellar2009} Rockafellar, T.J., Wets R.J.-B. (2009). \textit{Variational Analysis}. 3rd Edition, Springer.


\bibitem{WuHu1997}
Wu, H. and Huffer, F. W. (1997). Modeling the distribution of plant species using the autologistic
regression model. \textit{Environmental and Ecological Statistics} 4, 49--64.

\bibitem{Y1988} Younes, L. (1988).
Estimation and annealing for Gibbsian fields
Annales de l'I. H. P., sec.\ B, 24, no 2. 269--294.


\bibitem{Imput2010} Zalewska M., Niemiro W. and Samoli\'nski B. (2010). MCMC imputation in autologistic model.
\textit{Monte Carlo Methods Appl.} 16, 421--438.



\end{thebibliography}
\end{document}